\documentclass[aps,prl,superscriptaddress,twocolumn,floatfix,longbibliography,nobibnotes,nofootinbib]{revtex4-2}
\usepackage{dcolumn}
\usepackage{graphicx,color}
\usepackage{amsmath,amssymb,amsfonts,dsfont,mathrsfs,amsthm}
\usepackage{braket}
\usepackage{comment}
\usepackage{bm}
\usepackage[colorlinks=true,linkcolor=blue,citecolor=blue,urlcolor=blue]{hyperref}
\usepackage{soul}
\usepackage{bbm}
\usepackage{xspace}
\usepackage{verbatim}
\usepackage{amsthm}
\usepackage{color}
\usepackage{xcolor}
\usepackage{booktabs}
\usepackage{makecell}
\usepackage{overpic}
\usepackage{eso-pic}
\usepackage{enumerate}
\usepackage{physics}
\usepackage{bm}
\usepackage{ulem}
\hypersetup{colorlinks=true,linkcolor=blue,citecolor=blue, filecolor=blue,urlcolor=blue,breaklinks=true}

\graphicspath{{fig/}}

\newcommand*{\qfi}{\mathcal{F}_\mathrm{Q}}
\newcommand*{\cfi}{\mathcal{F}_\mathrm{C}}
\newcommand*{\avgcfi}{\mathcal{F}^\mathrm{avg}_\mathrm{C}}
\newcommand*{\qfim}{\boldsymbol{\mathcal{F}}_{\mathrm{Q}}}
\newcommand*{\cfim}{\boldsymbol{\mathcal{F}}_{\mathrm{C}}}
\newcommand*{\avgcfim}{\boldsymbol{\mathcal{F}}^\mathrm{avg}_\mathrm{C}}
\newcommand*{\variance}[1]{\mathrm{Var}(#1)}
\newcommand*{\cov}[1]{\mathrm{Cov}(#1)}
\newcommand*{\ii}{\mathrm{i}}
\newcommand*{\ee}{\mathrm{e}}

\newtheoremstyle{prlrunin}
  {\medskipamount}
  {\medskipamount}
  {\normalfont}
  {0pt}
  {\itshape}
  {.---}
  {0pt}
  {}
\theoremstyle{prlrunin}
\newtheorem{theorem}{Theorem}

\begin{document}

\title{Stabilizing temporal quantum-enhanced sensitivity via sub-optimal measurements}

\author{Wangsheng Zheng}
\email{wsh.zheng@std.uestc.edu.cn}
\affiliation{Institute of Fundamental and Frontier Sciences, University of Electronic Science and Technology of China, Chengdu 611731, China}

\author{Yaoling Yang}
\email{yyaoling@std.uestc.edu.cn}
\affiliation{Institute of Fundamental and Frontier Sciences, University of Electronic Science and Technology of China, Chengdu 611731, China}

\author{Abolfazl Bayat}
\email{abolfazl.bayat@uestc.edu.cn}
\affiliation{Institute of Fundamental and Frontier Sciences, University of Electronic Science and Technology of China, Chengdu 611731, China}
\affiliation{Key Laboratory of Quantum Physics and Photonic Quantum Information, Ministry of Education, University of Electronic Science and
Technology of China, Chengdu 611731, China}

\begin{abstract}
In non-equilibrium  probes, quantum-enhanced sensitivity is quantified through super-linear scaling of precision with respect to time, as a central metrology resource. However, this requires optimal measurements, which are often complex and time-dependent, making it challenging in practice. Sub-optimal measurements typically fail to retain robust super-linear scaling and show oscillations. Here, we begin by establishing a universal  temporal scaling law in a general multi-parameter non-equilibrium quantum sensing framework. Within this setting, we identify the universal behavior of the sensing precision with sub-optimal measurements which shows quadratic scaling modulated by an additional bounded oscillatory function. To remove the oscillatory part,  we propose a protocol in which measurements are partitioned  into different groups, each performed at different times. The collective precision obtained from this protocol  stabilizes a robust quadratic scaling even for sub-optimal measurements in the multi-parameter regime. We validate our protocol through three distinct examples as well as Bayesian estimation.  Our protocol is applicable to  every informative measurement and takes a key step towards practical realization of quantum-enhanced sensitivity.
\end{abstract}

\maketitle

\textit{Introduction.---} Quantum sensing represents one of the core areas that exemplifies the superiority of quantum technologies over their classical counterparts~\cite{degen2017quantum, braun2018quantum, ye2024essay, ghosh2026journey}. The estimation precision for inferring an unknown parameter $\theta$, quantified by the variance  $\variance{\hat{\theta}}$ of any locally unbiased estimator $\hat{\theta}$, is bounded by Cram{\'e}r-Rao inequality
$\variance{\hat{\theta}}\geq 1/\mathcal{M}\cfi(\theta)\geq 1/\mathcal{M}\qfi(\theta)$, where $\mathcal{M}$ indicates the number of independent repetitions, $\cfi(\theta)$ is Classical Fisher Information (CFI) determining the bound for a given measurement and $\qfi(\theta)$ is Quantum Fisher Information (QFI) which specifies the ultimate precision bound for an optimal measurement~\cite{Helstrom1969,Paris2009,liu2020quantum,meyer2021fisher}. In the absence of quantum resources, the Fisher information scales linearly with resources, e.g. time $t$, a constraint known as the standard quantum  limit. However, quantum features such as entanglement~\cite{giovannetti2004quantum,giovannetti2011advances,giovannetti2006quantum}, squeezing~\cite{pezze2018quantum,ma2011quantum,gessner2020multiparameter,heng2025quantum}, and criticality~\cite{campos2007quantum,zanardi2008quantum,invernizzi2008optimal,di2022multiparameter,di2023critical, Mukhopadhyay2024,alushi2024optimality,alushi2025collective,hotter2024combining,mihailescu2025uncertain,mihailescu2026critical,montenegro2025review} can boost the QFI scaling to a super-linear regime, known as quantum enhanced sensitivity.

When an unknown parameter is encoded in the quantum state of the probe through unitary evolution, the QFI generally scales as $\qfi{\sim} t^2$, indicating  quantum-enhanced sensitivity~\cite{Boixo2007,Pang2014}.
Non-equilibrium quantum sensing has been a subject of great interest in both theory~\cite{ravell2024strongly,mirkhalaf2024operational,Boixo2007,Pang2014,yang2025overcoming,yang2023extractable,aiache2024non,chu2021dynamic,yuan2017quantum,yousefjani2025discrete,yousefjani2025discretepra,gribben2025boundary,iemini2024floquet,gietka2022understanding,Lovett2013Differential,baak2024self,salvia2023critical,Puig2025,annabestani2022multiparameter,manshouri2025quantum,di2024metrology,mihailescu2025quantumsensing} and experiment~\cite{liu2021experimental,yu2025experimental,li2026nonequilibrium,xiao2026observation,Beaulieu2025Criticality,ding2022enhanced,liu2022deep,liu2601enhanced,tong2025topological,xiao2024non}.  In closed systems, despite the favorable quadratic temporal scaling of the QFI, achieving quantum-enhanced sensitivity in practice remains challenging and limited to small systems and short time scales~\cite{liu2021experimental,yu2025experimental,li2026nonequilibrium}. A key obstacle  is the necessity of using an optimal measurement basis that is often time-dependent and relies on highly entangled and non-local operations~\cite{Helstrom1969,Paris2009,liu2020quantum,meyer2021fisher}. For practically available sub-optimal measurements the CFI does not necessarily scale quadratically in time and often fluctuates significantly~\cite{li2026nonequilibrium}. Moreover, the fluctuations usually depend on the unknown parameter which makes it impossible to design  a sensing strategy that results in good precision for all range of values. Several open questions arise: (i) beyond the single-parameter sensing, can one identify universal scaling laws for multi-parameter sensing setups? (ii) is there a universal form for the scaling of CFI with sub-optimal measurements? and if so,  (iii) is it possible to stabilize quadratic temporal scaling even when measurement is sub-optimal?

In this Letter, we address the above questions. We first show that in a multi-parameter quantum sensing scenario, the eigenvalues of the QFI matrix exhibit temporal quadratic scaling, provided that a certain commonly satisfied condition holds. Under this condition, we demonstrate that for a sub-optimal measurement in the long-time limit, the precision bound dictated by the CFI matrix scales as ${\sim} f(\boldsymbol{\theta},t) t^2$, where $0{\le}f(\boldsymbol{\theta},t){\le} B$ is a bounded nonnegative function that generally oscillates in time. Such oscillations hinder robust sensing across the entire parameter space $\boldsymbol{\theta}$. Motivated by this observation, we propose a protocol in which measurements are divided into different groups each performed at different times. In this framework, we show that the Cram{\'e}r-Rao bound is governed by time-averaged CFIs, which effectively smooth out the oscillatory contributions and restore a stable $t^2$ scaling. The effectiveness of our protocol is demonstrated across three paradigmatic models. Finally, by exploiting Bayesian estimation we confirm that the precision tightly approaches the theoretical predicted  limit.

\textit{Estimation theory.---} A quantum probe state $\rho(\boldsymbol{\theta})$ encodes unknown parameters $\boldsymbol{\theta}{=}(\theta_1,\ldots,\theta_m)$ to be estimated. A measurement is described by a POVM $\{\Pi_x\}$ with outcome probabilities $p_x{\equiv}p(x|\boldsymbol{\theta}){=}\Trace[\rho(\boldsymbol{\theta})\Pi_x]$. The uncertainty, quantified by the covariance matrix $\cov{\hat{\boldsymbol{\theta}}}$ of any locally unbiased estimator $\hat{\boldsymbol{\theta}}$, satisfies the multi-parameter quantum Cram\'er-Rao theorem~\cite{Helstrom1969,Paris2009,liu2020quantum,meyer2021fisher} 
\begin{equation}\label{eq:mcrb}
\cov{\hat{\boldsymbol{\theta}}}\,\succeq\, \frac{\cfim(\boldsymbol{\theta})^{-1}}{\mathcal{M}}\,\succeq\, \frac{\qfim(\boldsymbol{\theta})^{-1}}{\mathcal{M}},
\end{equation}
where $\cfim(\boldsymbol{\theta})$ is the CFI matrix of the POVM, with elements $[\cfim(\boldsymbol{\theta})]_{ij}{=}\sum_x p_x^{-1}(\partial_{\theta_i} p_x)(\partial_{\theta_j} p_x)$, and $\qfim(\boldsymbol{\theta})$ is the QFI matrix, which upper bounds $\cfim(\boldsymbol{\theta})$ for every POVM~\cite{Braunstein1994,liu2020quantum}.  
One can use a positive-semidefinite weight matrix $\boldsymbol{\mathcal{W}}{\ge}0$ to get a scalar inequality as~\cite{Ragy2016}
\begin{equation}\label{eq:qcrb_scalar}
\Trace[\boldsymbol{\mathcal{W}}\cov{\hat{\boldsymbol{\theta}}}] {\geq } 
\frac{\Trace[\boldsymbol{\mathcal{W}}\cfim(\boldsymbol{\theta})^{-1}]}{\mathcal{M}} {\geq} 
\frac{\Trace[\boldsymbol{\mathcal{W}}\qfim(\boldsymbol{\theta})^{-1}]}{\mathcal{M}}.
\end{equation}
For the sake of notation simplicity, we define classical and quantum precision bounds $\mathscr{P}_{\mathrm{C}}$ and $\mathscr{P}_{\mathrm{Q}}$ which will be used throughout this paper for both single- and multi-parameter sensing. In the case of single-parameter sensing, $\mathscr{P}_{\mathrm{C}}{=}\cfi$ and $\mathscr{P}_{\mathrm{Q}}{=}\qfi$ while in the case of multi-parameter sensing 
$\mathscr{P}_{\mathrm{C}}{=}1/\Trace[\boldsymbol{\mathcal{W}}\cfim(\boldsymbol{\theta})^{-1}]$ and $\mathscr{P}_{\mathrm{Q}}{=}1/\Trace[\boldsymbol{\mathcal{W}}\qfim(\boldsymbol{\theta})^{-1}]$.

\textit{Non-equilibrium quantum sensing.---} Consider a quantum probe whose Hamiltonian $H(\boldsymbol{\theta}){=}\sum_\alpha E_\alpha P_\alpha$  depends on $m$ unknown parameters $\boldsymbol{\theta}$. Here, $P_\alpha$ denotes the projector onto the eigenspace of $H(\boldsymbol{\theta})$ corresponding to the eigenvalue $E_\alpha$. The parameters $\boldsymbol{\theta}$ are encoded into a quantum state via unitary evolution $\rho_{\boldsymbol{\theta}}(t)
{=}U_{\boldsymbol{\theta}}(t)\rho_0U_{\boldsymbol{\theta}}^\dagger(t)$, where $U_{\boldsymbol{\theta}}(t){=}\ee^{-\ii H(\boldsymbol{\theta})t}$ and $\rho_0$ is the initial state of the system which does not depend on $\boldsymbol{\theta}$.
In the case of a single-parameter sensing, i.e. $m{=}1$, the QFI is bounded by $\qfi{\le} t^2 \|\partial_\theta H(\theta)\|^2$, where $\|\cdot\|$ represents the difference between the largest and the smallest eigenvalues~\cite{Boixo2007,Pang2014,Puig2025}. 
The bound can be saturated for the Hamiltonians of the form $H(\theta){=}\theta H_0$ with an optimal initial state prepared in the equal-weight superposition of the eigenstates corresponding to extremal eigenvalues of $H_0$.
For the multi-parameter case, we put forward the following theorem.

\begin{theorem}\label{thm:scaling}
Let $H(\boldsymbol{\theta})$ be a finite-dimensional, time-independent Hamiltonian differentiable at the true parameter point, and let $\rho_0$ be a $\boldsymbol{\theta}$-independent initial state.
Assume $A_{\boldsymbol{u}}{:=}\partial_{\boldsymbol u}H(\boldsymbol{\theta})$ is the directional derivative of the Hamiltonian in the parameter space along a real unit vector $\boldsymbol{u}$, where $\partial_{\boldsymbol u}{:=}\boldsymbol{u}\vdot\grad_{\boldsymbol{\theta}}$.
Then every eigenvalue of $\qfim(\boldsymbol{\theta},t)$ grows quadratically in time if and only if
\begin{equation}\label{eq:gap-gradient-span}
\mathcal{S}\big[\rho_0,\mathcal{D}_H(A_{\boldsymbol{u}})\big]>0,\quad \forall \boldsymbol u\in\mathbb R^m,\ \|\boldsymbol{u}\|_2=1.
\end{equation}
Here, $\mathcal{D}_H(\cdot){=}\sum_\alpha P_\alpha(\cdot)P_\alpha$ denotes the pinching map with respect to $H(\boldsymbol{\theta})$ and $\mathcal{S}(\rho,\zeta){=}{-}\Tr([\sqrt{\rho},\zeta]^2)/2$ is the Wigner-Yanase skew information, which reduces to the variance of $\zeta$ for pure state $\rho$~\cite{Luo2003}.
\end{theorem}
\begin{proof}
    See Sec.~\ref{sec:SM_thm1} of the SM~\cite{SM}.
\end{proof}

The  above theorem is our first major result. The condition in Eq.~(\ref{eq:gap-gradient-span}) implies that $\rho_0$ and $\mathcal{D}_H(A_{\boldsymbol{u}})$ are noncommuting which is satisfied in most practical cases, as we will exemplify later.
While the QFI defines the ultimate precision limit, saturating it requires performing measurements in an optimal basis. In practice, however, such measurements may be prohibitively complex or, in the case of multi-parameter sensing, may not even exist due to the incompatibility issue~\cite{Ragy2016,Albarelli2020}. 
Therefore, it is highly desirable to investigate the precision  bounds for a generic sub-optimal measurement basis. 
The following theorem addresses this affirmatively.

\begin{theorem}\label{thm:two}
Consider the setting of Theorem~\ref{thm:scaling} such that every QFI eigenvalue grows quadratically in time. For any nonzero $\mathcal{W}{\succeq} 0$ and any fixed finite-outcome POVM that is informative about the parameters we get 
\begin{equation}\label{eq:fixed-povm-envelope}
\limsup_{t\to\infty}
\frac{1}{t^2\,\Trace[\boldsymbol{\mathcal W}\cfim(\boldsymbol{\theta},t)^{-1}]}>0 .
\end{equation}
\end{theorem}
\begin{proof}
    See Sec.~\ref{sec:SM_thm2} of the SM~\cite{SM}.
\end{proof}

The direct implication of the above theorems is 
\begin{equation} \label{eq:multiparams-bound}
f(\boldsymbol{\theta},t)t^2=\mathscr{P}_{\mathrm{C}}(\boldsymbol{\theta},t) \le \mathscr{P}_{\mathrm{Q}}(\boldsymbol{\theta},t) \le Bt^2
\end{equation}
where $B$ is a constant which reduces to $\|\partial_\theta H(\theta)\|^2$ for single-parameter sensing, and $0{\le} f(\boldsymbol{\theta},t){\le} B$ is a time-dependent bounded function (see Sec.~\ref{sec:SM_discussion_thm} of the SM~\cite{SM}). 
The equality $f(\boldsymbol{\theta},t)t^2{=}\mathscr{P}_{\mathrm{C}}(\boldsymbol{\theta},t)$ in Eq.~(\ref{eq:multiparams-bound}) is the second major result of our paper.  The key issue is that   $f(\boldsymbol{\theta},t)$ may show oscillatory behavior and take a value close to zero at the interrogation time, leading to  poor estimation precision and missing $t^2$ scaling. In the following, we provide a recipe to address this issue and restore quadratic temporal scaling  for any given informative measurement.

\textit{Stabilizing quadratic scaling.---} To guarantee temporal quadratic scaling for $\mathscr{P}_{\mathrm{C}}$ in Eq.~(\ref{eq:multiparams-bound}), one must stabilize $f(\boldsymbol{\theta},t)$ so that its time variation remains sufficiently slow, thereby rendering $t^2$ the dominant scaling term. In order to achieve this, we keep the measurement basis fixed and divide the total measurement sample $\mathcal{M}$ into $K$ independent groups, each performed at a different time $t_\nu$ and containing $\mathcal{M}{/}K$ measurements. Since the  measurements at different times are independent from each other, they generate $K$ independent probability distributions given by $\{ p_x^{(\nu)}(\boldsymbol{\theta})=\Trace[\Pi_x\rho_{\boldsymbol{\theta}}(t_\nu)]\}_{\nu=1}^{K}$.
Therefore, the probability of any outcome is determined as $p_{(x_1,x_2,\cdots,x_K)}=\prod_{\nu=1}^{K}p_{x_{\nu}}^{(\nu)}$ and thus its corresponding CFI matrix becomes $K\avgcfim(\boldsymbol{\theta})$ in which $\avgcfim(\boldsymbol{\theta})     {=} \frac{1}{K}\sum_{\nu=1}^{K}     \cfim^{(\nu)}(\boldsymbol{\theta})$. Substituting this into  Eq.~\eqref{eq:qcrb_scalar} gives the configuration-averaged Cram\'er-Rao bound
\begin{equation}\label{eq:avgcfim_crb}
\Trace[
\boldsymbol{\mathcal W}\,
\cov{\hat{\boldsymbol{\theta}}}
]
{\geq}
\frac{1}{\mathcal{M}}\Trace[
\boldsymbol{\mathcal W}\,
\avgcfim(\boldsymbol{\theta})^{-1}]{=}\frac{1}{\mathcal{M}\mathscr{P}_{\mathrm{C}}^{\rm avg}(\boldsymbol{\theta},t)}
\end{equation}
where $\mathscr{P}_{\mathrm{C}}^{\rm avg}(\boldsymbol{\theta},t){=}1/\Trace[
\boldsymbol{\mathcal W}\,
\avgcfim(\boldsymbol{\theta})^{-1}]$.
Eq.~(\ref{eq:avgcfim_crb}) is the third major result of our paper. 

As we demonstrate below, the time-averaged quantity $\mathscr{P}_{\mathrm{C}}^{\rm avg}(\boldsymbol{\theta},t)$ exhibits greater stability than any individual 
$\mathscr{P}_{\mathrm{C}}(\boldsymbol{\theta},t_\nu)$ and retains a robust $t^2$ scaling.  This increased robustness stems from the fact that oscillations or ill-conditioning at any given time can be compensated by data from other sampling instants. 
For the sake of simplicity and without loss of generality, we choose $K$ equally spaced sampling times $t_\nu=t{+}(2\nu-1-K)\tau/2$, with $\nu{=}1,\ldots,K$, which form an array centered at $t$ with spacing $\tau$ between consecutive samples. 
To make sure that the average sampling time is $t$, all sampling times must lie within $[0,2t]$. Therefore, the maximum number of time points would be $K_{\rm max}{=}\lfloor2t/\tau\rfloor{+}1$.
In the following, we illustrate this strategy on three models, covering both single- and multi-parameter estimation.

\begin{figure}[!b]
\centering
\includegraphics[width=1\linewidth]{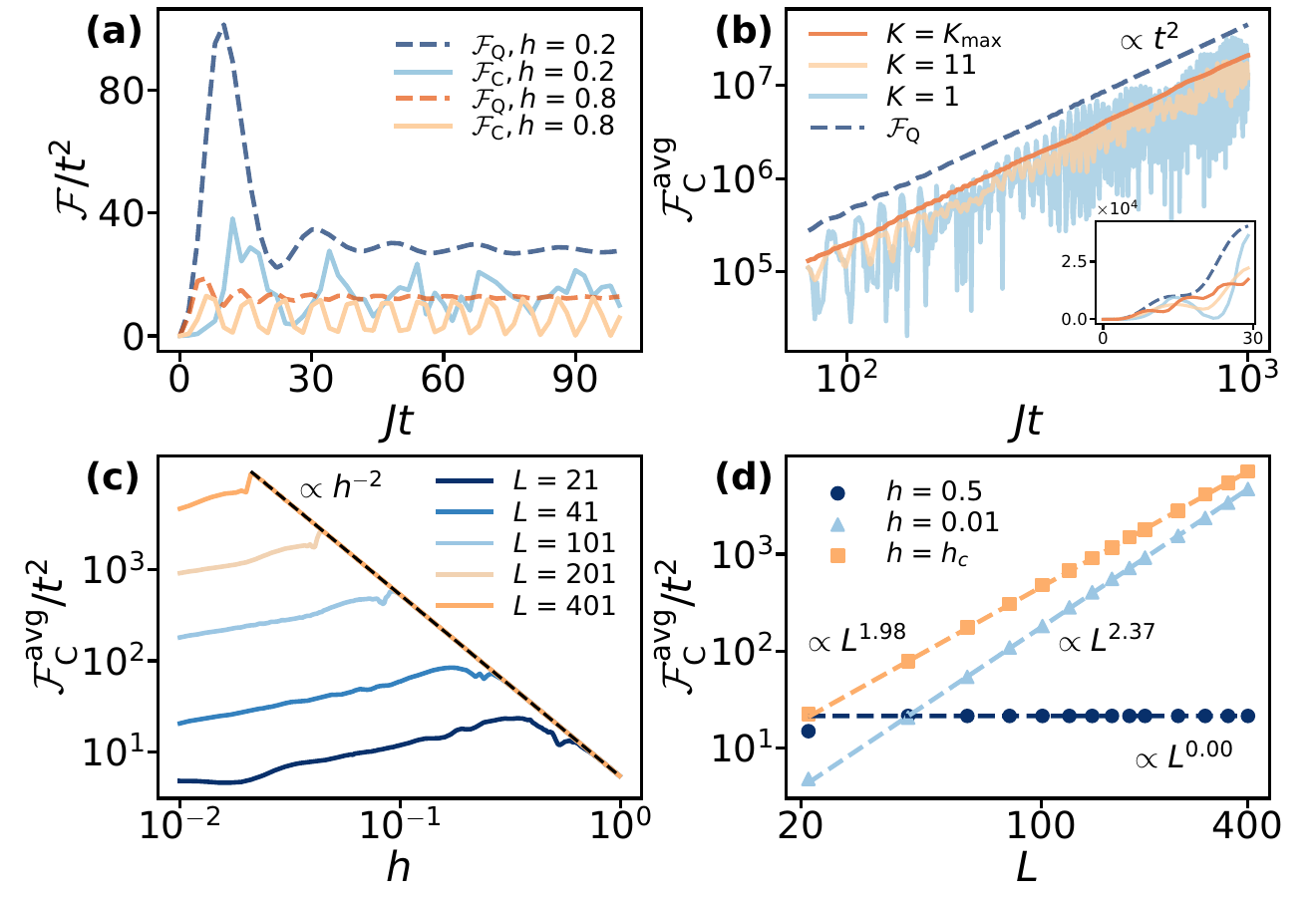}
\caption{Single-excitation Stark model. (a) $\qfi/t^2$ and $\cfi/t^2$ versus time for $h{=}0.2$ and $h{=}0.8$ with $L{=}21$. (b) The averaged CFI versus time for different $K$ at $h{=}0.4$, with $\tau{=}1$ and $L{=}21$. (c) $\avgcfi/t^2$ versus $h$ for various $L$ and (d) its scaling with $L$ across regimes, at $t{=}5000$, $\tau{=}1$, and $K{=}K_{\rm max}$.  \label{fig:stark}}
\end{figure}

\textit{Example 1: Single-excitation Stark chain.---}We first consider a one-dimensional chain of size $L$ subject to a linear gradient field $h$ to be estimated, described by the Hamiltonian
\begin{equation}\label{eq:stark}
H = -J\sum_{l=1}^{L-1}(\ket{l}\bra{l+1}+\ket{l+1}\bra{l})+h\sum_{l=1}^{L}l\ket{l}\bra{l},
\end{equation}
where $J{=}1$ is the exchange coupling, and $\ket{l}$ denotes the single-excitation state localized at site $l$. For finite $L$, this model exhibits a Stark-localization crossover around $h_c{\simeq}8J/L$, separating localized and extended regimes, across the entire spectrum~\cite{manshouri2025quantum,He2023}. We take $L$ to be odd and initialize the probe with a single excitation at the middle of the chain, $\ket{(L+1)/2}$. The measurement is a simple position measurement with $\{ \Pi_l=\ket{l}\bra{l}\}$.

We compute the QFI and the CFI for a chain of size $L{=}21$ in the extended ($h{=}0.2$) and localized ($h{=}0.8$) regimes.
As shown in Fig.~\ref{fig:stark}(a), both quantities, normalized by $t^2$, exhibit a transient peak at short times.
The normalized QFI then saturates to a constant, confirming the $t^2$ scaling, while the CFI oscillates below it.
Near the crossover, at $h{=}0.4$, Fig.~\ref{fig:stark}(b) shows the averaged CFI for several measurement groups $K$ with $\tau{=}1$ on a log-log scale. As $K$ increases, the curves become progressively smoother, with $K{=}1$ recovering the bare CFI. When $K{=}K_{\rm max}$ the averaged CFI grows parallel to the QFI and numerically recovers a smooth $t^2$ trend.
In Fig.~\ref{fig:stark}(c), the normalized averaged CFI is plotted as a function of $h$ for various system sizes considering $t{=}5000$, $\tau{=}1$, and $K{=}K_{\rm max}$. As the figure shows, the averaged CFI peaks around the crossover and converges to the $h^{-2}$ behavior in the localized phase, fully consistent with the QFI behavior studied in Ref.~\cite{manshouri2025quantum}.  In Fig.~\ref{fig:stark}(d), we plot the normalized averaged CFI as a function of system size $L$ for the same setting. The scaling confirms $\avgcfi\sim t^2 L^\beta$ in which $\beta \approx 2$ in the extended phase and $\beta \simeq 0$ in the localized phase. Indeed, a remarkable property of the Stark chain in the extended phase is its quantum-enhanced sensitivity with respect to both size and time.

\begin{figure}[!b]
\centering
\includegraphics[width=1\linewidth]{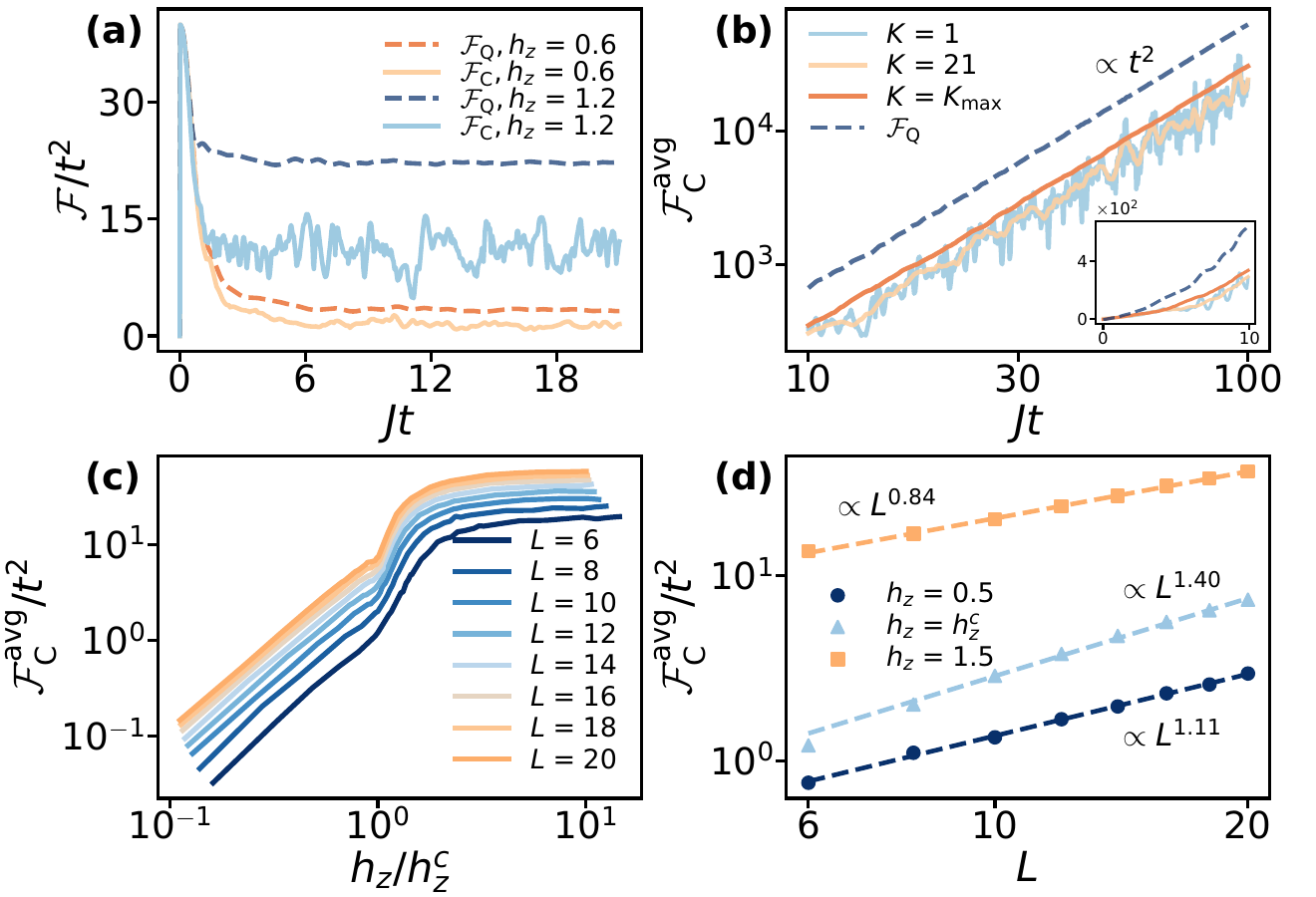}
\caption{Transverse-field Ising chain. (a) $\qfi/t^2$ and $\cfi/t^2$ versus time for $h_z{=}0.6$ and $h_z{=}1.2$ with $L{=}10$. (b) The averaged CFI versus time for different $K$ at $h_z{=}0.8$, with $\tau{=}0.1$ and $L{=}10$. (c) $\avgcfi/t^2$ versus $h_z/h_z^c$ for various $L$ and (d) its scaling with $L$ across regimes, at $t{=}100$, $\tau{=}1$, and $K{=}K_{\rm max}$. \label{fig:tfim}}
\end{figure}

\textit{Example 2: Transverse Ising chain.---} Consider an Ising chain of $L$ spin-$1/2$ particles  interacting through the Hamiltonian 
\begin{equation}\label{eq:tfim}
H=-J\sum_{l=1}^{L-1}\sigma^l_x\sigma^{l+1}_x-h_x\sum_{l=1}^L\sigma^l_x-h_z\sum_{l=1}^L\sigma^l_z.
\end{equation}
Here, $\sigma_\alpha^l$ $(\alpha{=}x,z)$ are the Pauli operators at site $l$, $J$  denotes the exchange coupling and is set to be 1, and $h_z$ and $h_x$ are transverse and longitudinal fields, respectively. We first focus on transverse Ising chain with $h_x{=}0$ and the goal is to estimate $h_z$. At zero temperature and in the thermodynamic limit ($L{\rightarrow}\infty$), the ground state of the system exhibits a quantum phase transition at the critical point ($h_z{=}h_z^c{=}{\pm} J$)~\cite{Pfeuty1970}. The many-body probe is initialized in the state $\ket{++\ldots+}$ and then freely evolves under the action of Hamiltonian $H$. Then local $\sigma_x$ measurements are subsequently performed on each spin.

The time evolution of the normalized QFI and CFI is shown in Fig.~\ref{fig:tfim}(a) for $L{=}10$ in the ferromagnetic ($h_z{=}0.6$) and paramagnetic ($h_z{=}1.2$) phases.
After transient peaks and oscillations, the normalized QFI saturates, confirming the $t^2$ scaling, and takes a larger value in the paramagnetic phase. In contrast, the normalized CFI oscillates below the QFI around a smaller mean. 
Near criticality, at $h_z{=}0.8$, Fig.~\ref{fig:tfim}(b) shows the performance of the protocol for several grouping numbers with $\tau{=}0.1$. The averaged CFI becomes progressively smoother as $K$ increases, and numerically recovers the quadratic trend when the time array spans twice the central time. 
The dependence of the averaged CFI on the field strength and the probe size is characterized in Figs.~\ref{fig:tfim}(c) and \ref{fig:tfim}(d), at $t{=}100$, $\tau{=}1$, and $K{=}K_{\rm max}$. In Fig.~\ref{fig:tfim}(c), the averaged CFI grows smoothly with $h_z/h_z^c$, where $h_z^c$ is the finite-size critical point determined by maximizing the ground-state QFI, and saturates deep in the paramagnetic phase. In Fig.~\ref{fig:tfim}(d), the averaged CFI follows the scaling $\avgcfi/t^2{\propto}L^\beta$ for $L{\in}[6,20]$, with $\beta{\approx}1.11$ at $h_z{=}0.5$, $\beta{\approx}1.40$ at $h_z{=}h_z^c$, and $\beta{\approx}0.84$ at $h_z{=}1.5$. Hence the exponent is largest at criticality. A Jordan-Wigner analysis in Sec.~\ref{sec:SM_spin_examples} of the SM~\cite{SM} indicates that the non-equilibrium  QFI shows quantum standard scaling, namely ${\sim} L$, in the thermodynamic limit. Since CFI is bounded by the QFI, one can conclude that the observed near-critical enhancement is a finite-size effect.  Nevertheless, the averaged CFI faithfully reproduces the dependence of the QFI on both $h_z$ and $L$.

\begin{figure}[!b]
\centering
\includegraphics[width=0.9\linewidth]{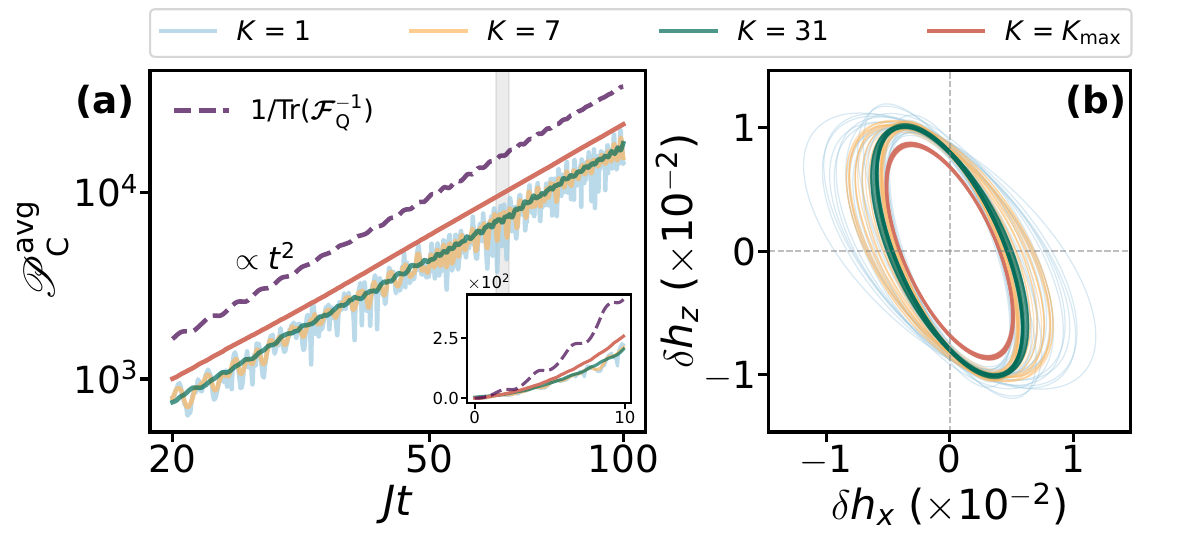}
\caption{Two-field Ising chain. (a) The quantity $\mathscr{P}_{\rm C}^{\rm avg}{=}1/\Trace[\avgcfim(\boldsymbol{h})^{-1}]$ versus time for $K{=}1,7,31,K_{\rm max}$, with $1/\Trace[\qfim^{-1}]$ (dashed) for reference, at $h_x{=}0.5$, $h_z{=}1.5$, $L{=}10$, and $\tau{=}0.1$. (b) Uncertainty ellipses $\delta\boldsymbol{h}^{\mathrm{T}}\avgcfim\,\delta\boldsymbol{h}{=}1$ in the $(\delta h_x,\delta h_z)$ plane, with one ellipse per time $t$ in the time window $Jt{\in}[63.5,66.5]$, highlighted in panel (a).
As $K$ increases, the ellipses corresponding to different times converge to each other and the area within the ellipse shrinks, indicating better sensitivity.}\label{fig:variance}
\end{figure}

\textit{Example 3: Two-field Ising chain.---} Here we consider both fields, namely $\boldsymbol{h}{=}(h_x,h_z)$, in the Hamiltonian (\ref{eq:tfim}) are unknown.  Thus, we deal with a two-parameter estimation problem. The initial state and the measurement basis remain the same. By considering the weight matrix $\boldsymbol{\mathcal{W}}{=}\mathbb{I}$, where $\mathbb{I}$ is the identity matrix, one gets 
$\variance{\hat h_x}{+}\variance{\hat h_z}{\geq}\mathcal{M}^{-1}\Trace[\avgcfim(\boldsymbol{h})^{-1}]$. The precision bound $\mathscr{P}_{\rm C}^{\rm avg}{=}1/\Trace[\avgcfim(\boldsymbol{h})^{-1}]$ is plotted in Fig.~\ref{fig:variance}(a) for $L{=}10$ and $\tau{=}0.1$. Increasing $K$ makes the curve smoother and brings it closer to the quadratic scaling indicated by $1/\Trace[\qfim^{-1}]$, showing that the time-array averaging also applies in the two-parameter setting.

For the two-parameter estimation of $h_x$ and $h_z$, we visualize $\avgcfim$ by the ellipse $\delta\boldsymbol{h}^{\mathsf T}\avgcfim\,\delta\boldsymbol{h}{=}1$~\cite{Johnson1998}. In the asymptotic limit $\cov{\hat{\boldsymbol h}}{\simeq}(\mathcal{M}\avgcfim)^{-1}$, this is the covariance contour of $\sqrt{\mathcal{M}}(\hat{\boldsymbol h}{-}\boldsymbol h)$, so a smaller area inside the ellipse indicates a higher joint precision, while the orientation identifies the best- and worst-determined parameter combinations (see Sec.~\ref{sec:SM_visualize_FIM} of the SM~\cite{SM}).
In Fig.~\ref{fig:variance}(b), each color corresponds to one $K$ and each curve of that color to a different interrogation time. For $K{=}1$ both the area and the orientation fluctuate from time to time, whereas for $K{=}K_{\rm max}$ the ellipses collapse onto a single, smaller one. The protocol thus mitigates the need to select a favorable interrogation time.

\begin{figure}[!b]
\centering
\includegraphics[width=1\linewidth]{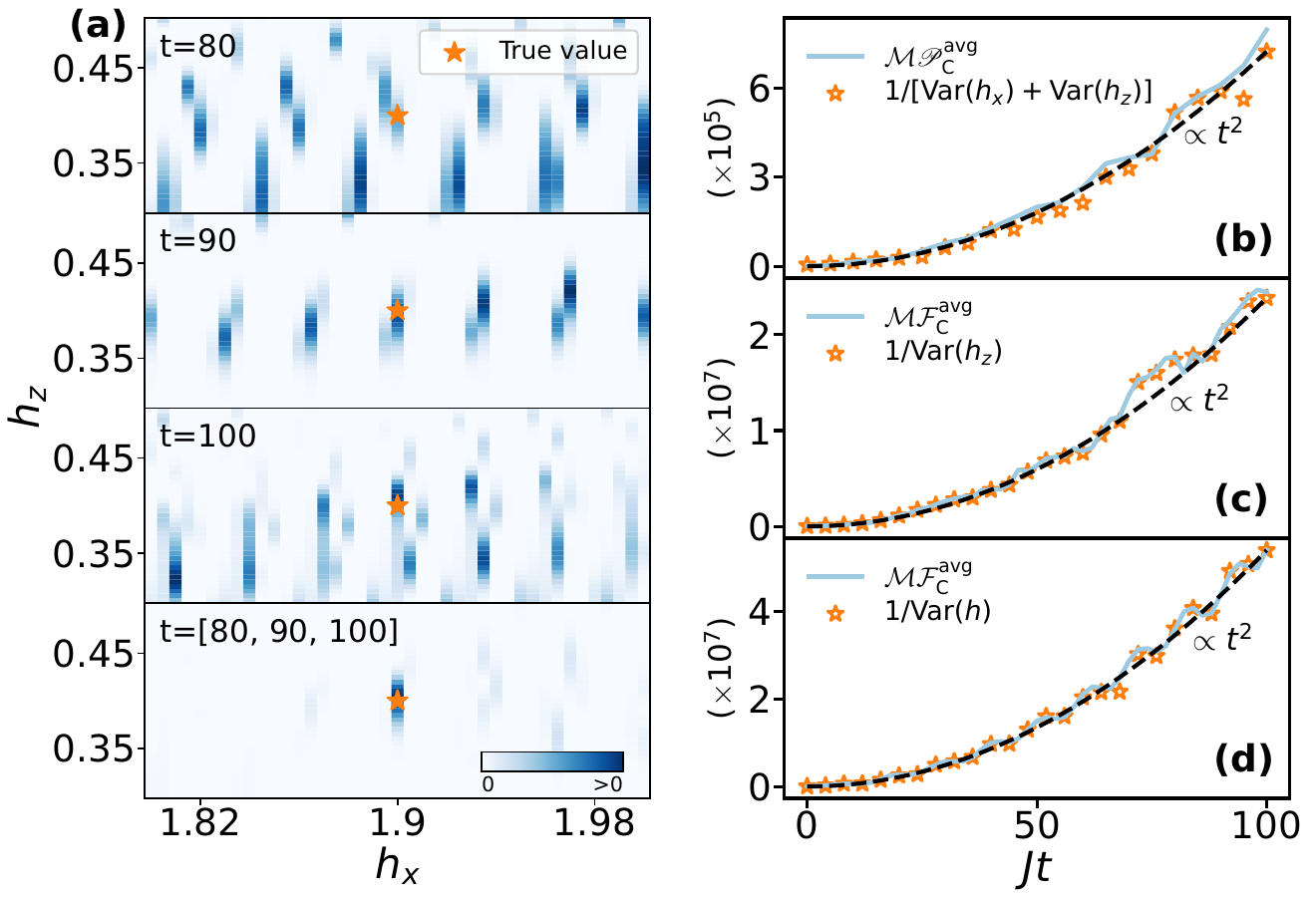}
\caption{Bayesian estimation for the time-array protocol. (a) Two-field Ising posteriors at $t{=}80,90,100$ and combined, for $h_x{=}1.9$, $h_z{=}0.4$, $\mathcal{M}{=}300$, and $L{=}6$. (b) $1/[\variance{h_x}+\variance{h_z}]$ against $\mathcal{M}\mathscr{P}_{\rm C}^{\rm avg}$ for the two-field Ising probe ($h_x{=}0.8$, $h_z{=}0.4$, $K{=}3$, $\tau{=}5$, $L{=}6$, $\mathcal{M}{=}3000$, and $200$ trials). Inverse Bayesian variance against $\mathcal{M}\avgcfi$ for (c) the transverse-field Ising chain ($L{=}6$, $h_z{=}1.0$) and (d) the single-excitation Stark probe ($L{=}21$, $h{=}0.6$), with $K{=}5$, $\tau{=}2$, $100$ samples per time, $\mathcal{M}{=}500$, and $200$ trials.\label{fig:ltfim}}
\end{figure}

\textit{Bayesian estimation.---} While the averaged CFI provides a bound on attainable precision, one has to employ a specific estimator, e.g. Bayesian inference,  to truly examine whether our protocol with grouping the measurements retains its $t^2$ scaling. Let $\mathcal D_\nu$ denote the data collected at interrogation time $t_\nu$. Starting from a prior $p_0(\boldsymbol{\theta})$, the Bayesian update is
\begin{equation}\label{eq:bayes_update}
p_\nu(\boldsymbol{\theta})
\propto p(\mathcal D_\nu|\boldsymbol{\theta},t_\nu)\,p_{\nu-1}(\boldsymbol{\theta}).
\end{equation}
Thus the posterior from one interrogation time becomes the prior for the next. The final posterior gives the estimator variance used below, see Sec.~\ref{sec:SM_bayesian} of the SM for details~\cite{SM}.

We first turn to the two-field Ising probe. As shown in Fig.~\ref{fig:ltfim}(a), for $L{=}6$, $\mathcal{M}=300$, $h_x{=}1.9$ and $h_z{=}0.4$, the posterior at a single interrogation time is multi-peaked, whereas combining the three times $t{=}80,90,100$ concentrates it around the true value. In Fig.~\ref{fig:ltfim}(b), for a two-field Ising probe of size $L{=}6$ with $(h_x,h_z){=}(0.8,0.4)$, $K{=}3$, and $\tau{=}5$, drawing $1000$ samples per interrogation time ($\mathcal{M}{=}3000$ in total) over $200$ independent trials, we plot $1/[\variance{h_x}+\variance{h_z}]$ and $\mathcal{M}\mathscr{P}_{\rm C}^{\rm avg}$ as a function of time. As the figure shows, $t^2$ scaling is achieved for both the averaged CFI matrix as well as the Bayesian estimation. In  Figs.~\ref{fig:ltfim}(c) and (d), we consider the single-parameter transverse-field Ising ($L{=}6$, $h_z{=}1.0$) and single-excitation Stark ($L{=}21$, $h{=}0.6$) probes with $K{=}5$, $\tau{=}2$, and $\mathcal{M}{=}500$ over $200$ independent trials. As the figures show, the inverse Bayesian variance as well as the rescaled averaged CFI show $t^2$ growth.

\textit{Conclusions.---} In this Letter, we show that in a multi-parameter non-equilibrium sensing setup quantum-enhanced sensitivity is manifested in quadratic scaling of the eigenvalues of the QFI matrix. In such a framework,  sub-optimal measurements, whose corresponding precision bound is governed by the CFI matrix, result in quadratic scaling modulated by a time-dependent oscillatory bounded function. We developed a protocol, based on dividing the measurements across various groups each performed at a different time, to overcome the oscillations and stabilize robust quadratic scaling of the precision for any informative but sub-optimal measurement. We show the effectiveness of the protocol through three examples as well as Bayesian estimation. Our protocol is general  and provides a key step forward for achieving quantum-enhanced sensitivity in existing non-equilibrium probes.

\textit{Acknowledgments.---} Authors acknowledge support from the National Natural Science Foundation of China (grants No.~W2541020, No.~12274059, No.~12574528, and No.~1251101297).

\bibliographystyle{apsrev4-2}

\bibliography{references}

\clearpage
\onecolumngrid
\AddToShipoutPictureFG*{
  \AtPageLowerLeft{
    \raisebox{0.62in}{
      \hspace*{0.67in}
      \begin{minipage}{0.8\textwidth}
        \footnotesize
        \rule{8em}{0.4pt}\par\vspace{0.45em}
        \makebox[1.2em][l]{\textsuperscript{*}}wsh.zheng@std.uestc.edu.cn\par
        \makebox[1.2em][l]{\textsuperscript{\ensuremath{\dagger}}}yyaoling@std.uestc.edu.cn\par
        \makebox[1.2em][l]{\textsuperscript{\ensuremath{\ddagger}}}abolfazl.bayat@uestc.edu.cn
      \end{minipage}
    }
  }
}
\begin{center}
{\large\bf Supplemental Material: Stabilizing temporal quantum-enhanced sensitivity via sub-optimal measurements\par}
\vspace{0.8em}
Wangsheng Zheng\textsuperscript{1,*},
Yaoling Yang\textsuperscript{1,\ensuremath{\dagger}},
and Abolfazl Bayat\textsuperscript{1,2,\ensuremath{\ddagger}}
\vspace{0.5em}

\textsuperscript{1}Institute of Fundamental and Frontier Sciences, University of Electronic Science and Technology of China, Chengdu 611731, China

\textsuperscript{2}Key Laboratory of Quantum Physics and Photonic Quantum Information, Ministry of Education, University of Electronic Science and Technology of China, Chengdu 611731, China
\end{center}
\vspace{1em}
\renewcommand{\theequation}{S\arabic{equation}}
\renewcommand{\thefigure}{S\arabic{figure}}
\renewcommand{\theHequation}{S\arabic{equation}}
\renewcommand{\theHfigure}{S\arabic{figure}}
\setcounter{equation}{0}
\setcounter{figure}{0}
\newcounter{smsection}
\renewcommand{\thesmsection}{\Roman{smsection}}
\providecommand{\theHsmsection}{}
\renewcommand{\theHsmsection}{SM.\Roman{smsection}}
\newcommand{\smsection}[2]{
\refstepcounter{smsection}
\label{#2}
\section*{\thesmsection.\ #1}
}

\theoremstyle{prlrunin}
\newtheorem*{definition}{Definition}

\smsection{Proof of Theorem 1}{sec:SM_thm1}

This section proves Theorem~\ref{thm:scaling} of the main text. We start by deriving the large-time form of the quantum Fisher information (QFI) matrix for a finite-dimensional, time-independent Hamiltonian. The derivation first analyzes the parameter derivative of the unitary evolution operator and isolates the part that grows linearly in time, which determines the leading term of the state derivative. We then characterize the condition under which the leading coefficient of the QFI matrix is positive definite, translating it into the Wigner-Yanase skew information criterion presented in the theorem.

The probe evolves as
\begin{equation}
\rho_{\boldsymbol{\theta}}(t)
=U_{\boldsymbol{\theta}}(t)\rho_0U_{\boldsymbol{\theta}}^\dagger(t),
\qquad
U_{\boldsymbol{\theta}}(t)=\mathrm{e}^{-\ii H(\boldsymbol{\theta})t},
\end{equation}
where $\boldsymbol{\theta}{=}(\theta_1,\theta_2,\ldots,\theta_m)$ denotes the unknown parameters to be estimated. The initial state $\rho_0$ is independent of $\boldsymbol{\theta}$, and all quantities are evaluated at the true parameter point. Let
\begin{equation}
H=\sum_\alpha E_\alpha P_\alpha,
\qquad
A_i:=\partial_{\theta_i}H ,
\end{equation}
where $P_\alpha$ is the spectral projector onto the eigenspace with energy $E_\alpha$. For each parameter, we compare $\partial_{\theta_i}U_{\boldsymbol{\theta}}(t)$ with $U_{\boldsymbol{\theta}}(t)$ itself and define the parameter generator
\begin{equation}
\mathcal K_i(t):=
\ii U_{\boldsymbol{\theta}}^\dagger(t)\partial_{\theta_i}U_{\boldsymbol{\theta}}(t)
=\int_0^t \mathrm{e}^{\ii vH}A_i\mathrm{e}^{-\ii vH}\,\dd v .
\end{equation}
This equality follows from the Wilcox parameter-derivative formula for exponential operators~\cite{Wilcox1967,Liu2015}.
Inserting the spectral decomposition and integrating each block yields
\begin{equation}\label{eq:SM_K_decomposition}
\mathcal K_i(t)=tG_i+R_i(t),
\qquad
G_i:=\sum_\alpha P_\alpha A_iP_\alpha=\mathcal{D}_H(A_i) ,
\end{equation}
where the off-diagonal blocks
\begin{equation}
R_i(t)=
\sum_{\alpha\ne\beta}
\frac{\mathrm{e}^{\ii (E_\alpha-E_\beta)t}-1}{\ii (E_\alpha-E_\beta)}
P_\alpha A_iP_\beta
\end{equation}
oscillate without growing. Since the Hilbert space is finite dimensional, the finitely many gaps $E_\alpha{-}E_\beta$ are bounded away from zero, ensuring $R_i(t){=}O(1)$ uniformly in $t$. Hence, only the block-diagonal parts $G_i$ accumulate linearly in time, and
\begin{equation}\label{eq:SM_state_derivative}
\partial_{\theta_i}\rho_{\boldsymbol{\theta}}(t)=-\ii U_{\boldsymbol{\theta}}(t)[\mathcal K_i(t),\rho_0]U_{\boldsymbol{\theta}}^\dagger(t)
=t\,X_i(t)+O(1).
\end{equation}
The operator $X_i(t)$ in Eq.~\eqref{eq:SM_state_derivative} is defined as
\begin{equation}
X_i(t):=-\ii U_{\boldsymbol{\theta}}(t)[G_i,\rho_0]U_{\boldsymbol{\theta}}^\dagger(t)=-\ii [\mathcal{D}_H(A_i),\rho_{\boldsymbol{\theta}}(t)].
\end{equation}
Since both $G_i$ and $\rho_{\boldsymbol{\theta}}(t)$ are Hermitian, $X_i(t)$ is Hermitian, and it is traceless because the trace of a commutator vanishes. These operators govern the entire large-time analysis and determine the asymptotic QFI matrix.

We now evaluate the symmetric logarithmic derivative (SLD) QFI formula in the eigenbasis $U_{\boldsymbol{\theta}}(t)\ket{r}$ of $\rho_{\boldsymbol{\theta}}(t)$, where $\rho_0{=}\sum_r\lambda_r\ket{r}\bra{r}$.
The elements of the QFI matrix are given by~\cite{liu2020quantum}
\begin{equation}\label{eq:SM_qfim_generator}
[\qfim(\boldsymbol{\theta},t)]_{ij}
=2\sum_{r,s:\lambda_r+\lambda_s>0}
\frac{
\Re\!\left[
\bra{r}U_{\boldsymbol{\theta}}^\dagger(t)\partial_{\theta_i}\rho_{\boldsymbol{\theta}}(t)U_{\boldsymbol{\theta}}(t)\ket{s}
\bra{s}U_{\boldsymbol{\theta}}^\dagger(t)\partial_{\theta_j}\rho_{\boldsymbol{\theta}}(t)U_{\boldsymbol{\theta}}(t)\ket{r}
\right]
}{\lambda_r+\lambda_s}.
\end{equation}
The eigenvalues ${\lambda_r}$ do not depend on $t$, so the denominators remain constant. Substituting Eq.~\eqref{eq:SM_state_derivative} into Eq.~\eqref{eq:SM_qfim_generator} gives
\begin{equation}\label{eq:SM_qfim_asymptotic}
\qfim(\boldsymbol{\theta},t)=t^2\vb{\Gamma}_{\rm Q}+O(t),
\end{equation}
where
\begin{equation}\label{eq:SM_GammaQ}
[\vb{\Gamma}_{\rm Q}]_{ij}
=2\sum_{r,s:\lambda_r+\lambda_s>0}
\frac{
\Re\!\left[
\bra{r}X_i(0)\ket{s}\bra{s}X_j(0)\ket{r}
\right]
}{\lambda_r+\lambda_s}.
\end{equation}
Since $\qfim/t^2{\to}\vb{\Gamma}_{\rm Q}$ as $t {\to} \infty$, the eigenvalues of $\qfim/t^2$ converge to the eigenvalues of $\vb{\Gamma}_{\rm Q}$. Hence, every QFI eigenvalue has a positive $t^2$ coefficient if and only if $\vb{\Gamma}_{\rm Q}{\succ}0$.

For any real unit vector $\boldsymbol{u}{=}(u_1,\ldots,u_m)\in\mathbb R^m$, with $m$ the number of parameters, we have
\begin{equation}\label{eq:SM_directional_gamma}
\boldsymbol u^{\mathsf T}\vb{\Gamma}_{\rm Q}\boldsymbol u
=2\sum_{r,s:\lambda_r+\lambda_s>0}
\frac{|\bra{r}X_{\boldsymbol u}(0)\ket{s}|^2}{\lambda_r+\lambda_s},
\end{equation}
where
\begin{equation}
X_{\boldsymbol u}(t):=\sum_i u_iX_i(t)=\sum_i -\ii[\mathcal{D}_H(u_iA_i),\rho_{\boldsymbol{\theta}}(t)].
\end{equation}
Eq.~\eqref{eq:SM_directional_gamma} vanishes if and only if $X_{\boldsymbol u}(0){=}0$. Therefore, the presence of the quadratic term in Eq.~\eqref{eq:SM_qfim_asymptotic} necessitates that the commutator between the pinched directional derivative and the initial state is strictly non-zero.
This degree of noncommutativity is exactly quantified by the Wigner-Yanase skew information.
Recalling that $A_{\boldsymbol{u}}{:=}\partial_{\boldsymbol u}H{=}\sum_i u_iA_i$, we obtain the criterion
\begin{equation}\label{eq:SM_positive_condition}
\vb{\Gamma}_{\rm Q}\succ0\ \Longleftrightarrow\ -\ii[\mathcal{D}_H(A_{\boldsymbol{u}}),\rho_0]\neq 0
\ \Longleftrightarrow\ 
\mathcal{S}\big[\rho_0,\mathcal{D}_H(A_{\boldsymbol{u}})\big]>0, \quad \forall \boldsymbol u\in\mathbb R^m,\ \|\boldsymbol{u}\|=1.
\end{equation}
This completes the proof of Theorem~\ref{thm:scaling}.\hfill$\square$

\smsection{Proof of Theorem 2}{sec:SM_thm2}

This section proves Theorem~\ref{thm:two} of the main text. We first introduce a sufficient informativeness condition under which the quadratic envelope can be obtained.
Consider a positive operator-valued measure (POVM) $\{\Pi_x\}_{x=1}^{N}$ with outcome probabilities $p_x(t){=}\Trace[\Pi_x\rho_{\boldsymbol{\theta}}(t)]$. The informativeness is defined as below.

\begin{definition}
A finite-outcome POVM $\{\Pi_x\}_{x=1}^{N}$ is \emph{informative} about the parameters if there exists a time $t_\star{\geq} 0$ at which, for every real unit direction $\boldsymbol u$, at least one outcome satisfies
\begin{equation}\label{eq:SM_informative_povm}
\Trace[\Pi_xX_{\boldsymbol u}(t_\star)]\ne0 .
\end{equation}
\end{definition}
For the vector $\boldsymbol v_x$ with $[\boldsymbol v_x]_i {:=} \Trace[\Pi_xX_i(t_\star)]$ ($i{=}1,\ldots,m$), linearity gives $\Trace[\Pi_xX_{\boldsymbol u}(t_\star)]{=}\sum_i u_i\Trace[\Pi_xX_i(t_\star)]{=}\boldsymbol u^{\mathsf T}\boldsymbol v_x$.
Informativeness therefore rules out any nonzero vector orthogonal to all $\boldsymbol v_x$, which is equivalent to the set of vectors $\{\boldsymbol v_x\}_{x=1}^N$ fully spanning $\mathbb R^m$. Since $\sum_x\Pi_x{=}\mathbb I$ and each $X_i(t_\star)$ is traceless, these vectors strictly sum to zero, restricting their span to a subspace of dimension at most $N{-}1$. 
Hence, an informative POVM must have $N{\geq}m{+}1$ outcomes, aligning with the invertibility condition for the CFI matrix given in Ref.~\cite{Candeloro2024}.

Since the system is finite-dimensional, the unitary evolution has the recurrence property: for every $\epsilon{>}0$ and every $T{>}0$, there is a time $t{>}T$ at which all relative phases are within $\epsilon$ of their initial values,
\begin{equation}
|\ee^{-\ii (E_\alpha-E_\beta)t}-1|<\epsilon
\quad\text{for all }\alpha,\beta ,
\end{equation}
where $E_\alpha$ and $E_\beta$ are eigenvalues of the parameter-dependent Hamiltonian $H(\boldsymbol\theta)$.
In particular, $U_{\boldsymbol{\theta}}(t)$ is within $\epsilon$, in operator norm, of a global phase times the identity~\cite{Wallace2015}.
Setting $\epsilon{=}1/n$ and $T{=}n$, we obtain a sequence $\{t_n^\prime\}_{n=1}^\infty$ along which every relative phase $\ee^{-\ii (E_\alpha-E_\beta)t_n^\prime}$ tends to $1$. Consequently, as $n{\to}\infty$,
\begin{equation}
U_{\boldsymbol{\theta}}(t_n^\prime)\rho_{\boldsymbol{\theta}}(t_\star)U_{\boldsymbol{\theta}}^\dagger(t_n^\prime)\to\rho_{\boldsymbol{\theta}}(t_\star),
\qquad
U_{\boldsymbol{\theta}}(t_n^\prime)X_i(t_\star)U_{\boldsymbol{\theta}}^\dagger(t_n^\prime)\to X_i(t_\star) ,
\end{equation}
so $p_x(t_n){\to} p_x(t_\star)$ with $t_n{=}t_n^\prime{+}t_\star$, and Eq.~\eqref{eq:SM_state_derivative} gives
\begin{equation}\label{eq:SM_dpx_convergence}
\frac{1}{t_n}\partial_{\theta_i}p_x(t_n)
=\Trace[\Pi_xU_{\boldsymbol{\theta}}(t_n^\prime)X_i(t_\star)U_{\boldsymbol{\theta}}^\dagger(t_n^\prime)]+o(1)
\to \Trace[\Pi_xX_i(t_\star)] .
\end{equation}

The CFI matrix $[\cfim]_{ij}{=}\sum_xp_x^{-1}(\partial_{\theta_i}p_x)(\partial_{\theta_j}p_x)$ is a sum of positive-semidefinite contributions, one for each outcome with $p_x{>}0$.
For every
$x$ such that $p_x(t_\star){>}0$, the convergence $p_x(t_n){\to}p_x(t_\star)$ implies that
$p_x(t_n){>}0$ for all sufficiently large $n$. Since $p_x(t_n){\leq}1$ and every CFI contribution is positive semidefinite, for all such $n$,
\begin{equation}
\frac{\cfim(\boldsymbol\theta,t_n)}{t_n^2}
=\sum_{x:\,p_x(t_n)>0}\frac{1}{p_x(t_n)}\frac{[\grad_{\boldsymbol{\theta}}p_x(t_n)][\grad_{\boldsymbol{\theta}}p_x(t_n)]^{\mathsf T}}{t_n^2}\succeq\sum_{x:\,p_x(t_\star)>0}\frac{[\grad_{\boldsymbol{\theta}}p_x(t_n)][\grad_{\boldsymbol{\theta}}p_x(t_n)]^{\mathsf T}}{t_n^2}.
\end{equation}
We correspondingly define a matrix $\vb{\Lambda}$ with elements
\begin{equation}
\vb{\Lambda}_{ij}:=\sum_{x:\,p_x(t_\star)>0}
\Trace[\Pi_xX_i(t_\star)]\Trace[\Pi_xX_j(t_\star)].
\end{equation}
The matrix $\vb{\Lambda}$ is in fact positive definite. 
Since both $\Pi_x$ and $\rho_{\boldsymbol{\theta}}(t_\star)$ are positive semi-definite, $\Pi_x\rho_{\boldsymbol{\theta}}(t_\star){=}\rho_{\boldsymbol{\theta}}(t_\star)\Pi_x{=}0$ whenever $p_x({t_\star}){=}0$, resulting in $\Trace[\Pi_xX_{\boldsymbol u}(t_\star)]{=}0$.
Informativeness therefore ensures that at least one outcome satisfies 
$\Trace[\Pi_xX_{\boldsymbol u}(t_\star)]{\neq}0$ with $p_x({t_\star}){>}0$. Hence, for any real unit vector $\boldsymbol{u}$,
\begin{equation}
\boldsymbol u^{\mathsf T}\vb{\Lambda}\boldsymbol u
=\sum_{x:\,p_x(t_\star)>0}
\left|\Trace[\Pi_xX_{\boldsymbol u}(t_\star)]\right|^2>0,
\end{equation}
and consequently $\vb{\Lambda}{\succ}0$.

Following from Eq.~\eqref{eq:SM_dpx_convergence} and the positive-definiteness of $\vb{\Lambda}$, we have
\begin{equation}
\sum_{x:\,p_x(t_\star)>0}\frac{[\grad_{\boldsymbol{\theta}}p_x(t_n)][\grad_{\boldsymbol{\theta}}p_x(t_n)]^{\mathsf T}}{t_n^2}\succeq \vb{\Lambda}+o(1)\mathbb{I}\succeq \frac{1}{2}\vb{\Lambda} .
\end{equation}
Therefore, for all sufficiently large $n$,
\begin{equation}
\cfim(\boldsymbol{\theta},t_n)\succeq \frac{t_n^2}{2}\vb{\Lambda} .
\end{equation}
Since both sides are positive definite, inversion reverses the matrix inequality:
$\cfim(\boldsymbol\theta,t_n)^{-1}{\preceq}2\vb{\Lambda}^{-1}/t_n^2$.
Multiplying by $\boldsymbol{\mathcal W}{\succeq}0$, taking the trace, and then taking reciprocals gives
\begin{equation}
\frac{1}{\Trace[\boldsymbol{\mathcal W}\cfim(\boldsymbol\theta,t_n)^{-1}]}
\ge
\frac{t_n^2}{2\Trace[\boldsymbol{\mathcal W}\vb{\Lambda}^{-1}]}.
\end{equation}
The denominator is finite and strictly positive because
$\vb{\Lambda}^{-1}{\succ}0$ and $\boldsymbol{\mathcal W}{\succeq}0$ is nonzero. Consequently, the limit superior in Eq.~\eqref{eq:fixed-povm-envelope} is at least
$1/(2\Trace[\boldsymbol{\mathcal W}\vb{\Lambda}^{-1}]){>}0$, which completes the proof of Theorem~\ref{thm:two}.\hfill$\square$

\smsection{QUADRATIC UPPER BOUNDS AND GENERICITY OF INFORMATIVE POVMS}{sec:SM_discussion_thm}

In this section, we derive two further results related to Theorems~\ref{thm:scaling} and~\ref{thm:two}. First, we obtain a quadratic upper bound on the quantum and classical precision measures, thereby justifying the bounded coefficient $0{\leq}f(\boldsymbol\theta,t){\leq}B$ used in the main text. Second, under the quadratic-scaling condition of Theorem~\ref{thm:scaling}, we show that, at any fixed time and true parameter point, informative finite-outcome POVMs with $N{\geq}m{+}1$ form an open, dense, full-measure subset of the $N$-outcome POVM space.

\subsection{A. QUADRATIC UPPER BOUND}

The quadratic upper bound follows by bounding the QFI in each parameter direction. Recall that the notation
$\|G\|{:=}\lambda_{\max}(G)-\lambda_{\min}(G)$ denotes the spectral range of a Hermitian operator $G$. For a real unit vector $\boldsymbol u$, the local generator is
\begin{equation}
\mathcal K_{\boldsymbol u}(t)
:=\ii U_{\boldsymbol{\theta}}^\dagger(t)
\partial_{\boldsymbol u}U_{\boldsymbol{\theta}}(t)
=\sum_i u_i\mathcal K_i(t)
\end{equation}
With the spectral decomposition $\rho_0{=}\sum_r\lambda_r\ket{r}\bra{r}$, the QFI in this direction is bounded by the variance of the local generator~\cite{Pang2014,liu2020quantum}:
\begin{align}
\boldsymbol u^{\mathsf T}\qfim(\boldsymbol{\theta},t)\boldsymbol u
&\leq \sum_{r} 4\lambda_r \operatorname{Var}_{\ket{r}}(\mathcal{K}_{\boldsymbol{u}}(t)) \nonumber\\
&\leq4\operatorname{Var}_{\rho_0}
\!(\mathcal{K}_{\boldsymbol{u}}(t)) \nonumber\\
&\le
\|\mathcal K_{\boldsymbol u}(t)\|^2 .
\end{align}
The second inequality follows from Theorem~1 in Ref.~\cite{Toth2013}. The third uses
$\operatorname{Var}_{\rho}(G){\leq}\|G\|^2/4$. The integral representation of the generator then bounds its spectral range:
\begin{align}
\|\mathcal K_{\boldsymbol u}(t)\|
&=\left\|
\int_0^t\mathrm{e}^{\ii vH}A_{\boldsymbol u}
\mathrm{e}^{-\ii vH}\,\dd v
\right\|\nonumber\\
&\le t\|A_{\boldsymbol u}\|
\le t\sum_i|u_i|\|A_i\|
\le t\left(\sum_i\|A_i\|^2\right)^{1/2},
\end{align}
The final step uses the Cauchy--Schwarz inequality and
$\|\boldsymbol u\|_2{=}1$. The resulting bound is uniform over all real unit directions:
$\boldsymbol u^{\mathsf T}(c_{\rm Q}t^2\mathbb I-\qfim(\boldsymbol{\theta},t))\boldsymbol u
{\geq}0$, where $c_{\rm Q}{:=}\sum_i\|A_i\|^2$. It therefore gives the matrix inequality
\begin{equation}
\qfim(\boldsymbol{\theta},t)
\preceq c_{\rm Q}t^2\,\mathbb I.
\end{equation}

This matrix bound also controls the inverse-trace precision. When
$\cfim(\boldsymbol{\theta},t)$ is nonsingular, the Braunstein--Caves inequality
$\cfim{\preceq}\qfim$~\cite{Braunstein1994} ensures that $\qfim$ is nonsingular as well.
Inversion reverses the order of positive-definite matrices, so the two bounds combine to give, for $t{>}0$,
\begin{equation}
\cfim(\boldsymbol{\theta},t)^{-1}
\succeq\qfim(\boldsymbol{\theta},t)^{-1}
\succeq\frac{1}{c_{\rm Q}t^2}\,\mathbb I .
\end{equation}
Taking the trace with the positive-semidefinite weight
$\boldsymbol{\mathcal W}{\succeq}0$ preserves the inequalities:
\begin{equation}
\Trace[\boldsymbol{\mathcal W}\cfim(\boldsymbol{\theta},t)^{-1}]
\ge
\Trace[\boldsymbol{\mathcal W}\qfim(\boldsymbol{\theta},t)^{-1}]
\ge
\frac{1}{Bt^2}.
\end{equation}
Here $B{:=}c_{\rm Q}/\Trace[\boldsymbol{\mathcal W}]$ is independent of time. In the nonsingular case considered here, $c_{\rm Q}{>}0$. The nonzero positive-semidefinite weight also satisfies $\Trace[\boldsymbol{\mathcal W}]{>}0$. Thus $B{>}0$. Taking reciprocals yields the precision bounds
\begin{equation}\label{eq:multiparams-bound-SM}
\frac{1}{\Trace[\boldsymbol{\mathcal W}\cfim(\boldsymbol{\theta},t)^{-1}]}
\le
\frac{1}{\Trace[\boldsymbol{\mathcal W}\qfim(\boldsymbol{\theta},t)^{-1}]}
\le
Bt^2.
\end{equation}
The singular cases follow from the convention that the inverse-trace precision is zero whenever the Fisher information matrix is singular. If $\cfim$ is singular but $\qfim$ is nonsingular, the classical precision vanishes, while inversion of $\qfim(\boldsymbol\theta,t){\preceq}c_{\rm Q}t^2\mathbb I$ gives the same quantum upper bound $Bt^2$. If $\qfim$ is singular, the Braunstein--Caves inequality forces $\cfim$ to be singular as well, and both precision terms vanish. Equation~\eqref{eq:multiparams-bound-SM} therefore holds in all cases.
For single-parameter sensing, the main-text definition uses the scalar weight $\boldsymbol{\mathcal W}{=}1$, giving $B{=}\|\partial_\theta H(\theta)\|^2$. Equation~\eqref{eq:multiparams-bound-SM} then reduces to $\cfi\leq\qfi\leq t^2\|\partial_\theta H(\theta)\|^2$, in agreement with Refs.~\cite{Boixo2007,Pang2014,Puig2025}.

The precision bound thus establishes the bounded coefficient introduced after Theorem~\ref{thm:two}:
\begin{equation}
f(\boldsymbol{\theta},t)
:=\frac{1}{t^2\Trace[\boldsymbol{\mathcal W}
\cfim(\boldsymbol{\theta},t)^{-1}]},
\qquad
0\leq f(\boldsymbol{\theta},t)
\leq B.
\end{equation}
For an informative fixed measurement, Theorem~\ref{thm:two} also gives
$\limsup_{t\to\infty}f(\boldsymbol{\theta},t){>}0$.
Thus every fixed measurement obeys a quadratic upper bound on precision, and an informative fixed measurement attains a positive quadratic envelope at arbitrarily large interrogation times.

\subsection{B. GENERICITY OF INFORMATIVE POVMS}

We now show that the informativeness condition in Sec.~\ref{sec:SM_thm2} does not depend on a special choice of measurement. Under the condition of Theorem~\ref{thm:scaling}, and at a fixed true parameter point, the POVMs that are informative at some time form a relatively open, dense, full-measure subset of the $N$-outcome POVM space for every $N{\geq}m{+}1$.

To prove this, first fix an arbitrary time $t{\geq}0$, and let $\mathcal G_t$ denote the POVMs that are informative at that time. The condition of Theorem~\ref{thm:scaling} ensures that the traceless Hermitian operators $X_1(t),\ldots,X_m(t)$ are linearly independent over $\mathbb R$. Since $\Trace[\Pi_xX_i(t)]{=}0$ whenever $p_x(t){=}0$, the restricted sum defining $\vb{\Lambda}(t)$ equals the unrestricted sum over all outcomes. Hence, $\det\vb{\Lambda}(t)$ is a polynomial in the POVM elements and is strictly positive on $\mathcal G_t$.

This polynomial is not identically zero. Indeed, take the POVM elements $\mathbb I/N{+}\epsilon X_i(t)$ for $i{=}1,\ldots,m$, the element $\mathbb I/N{-}\epsilon\sum_{i=1}^{m}X_i(t)$, and any remaining elements equal to $\mathbb I/N$. They sum to the identity and are positive definite for sufficiently small $\epsilon{>}0$. For the first $m$ outcomes,
\begin{equation*}
\Trace[\Pi_xX_i(t)]
=\epsilon\,\Trace[X_x(t)X_i(t)],
\qquad x,i=1,\ldots,m .
\end{equation*}
The matrix $[\Trace(\Pi_xX_i(t))]_{x,i=1}^{m}$ is therefore $\epsilon$ times the nonsingular Hilbert--Schmidt Gram matrix of the $X_i(t)$. Hence, the constructed POVM is informative at time $t$.

The zero set of a nonzero polynomial has Lebesgue measure zero and empty interior~\cite{Mityagin2015}. The POVM set is convex and contains $\Pi_x=\mathbb I/N$ as an interior point relative to the affine space $\sum_x\Pi_x=\mathbb I$. It follows that $\mathcal G_t$ is relatively open, dense, and of full measure in the POVM set, with measure understood relative to this affine space.

Finally, the original definition requires informativeness only at some time. Thus, the full set of informative POVMs is
\begin{equation*}
\mathcal G=\bigcup_{t\geq0}\mathcal G_t .
\end{equation*}
This union is relatively open, and it is dense and of full measure because it contains any one of the sets $\mathcal G_t$. Thus, almost every POVM with sufficiently many outcomes is informative, and informativeness persists under sufficiently small perturbations within the POVM space.

\smsection{Supporting data for spin-chain examples}{sec:SM_spin_examples}

Here, we compute the QFI of the ground state for finite transverse-field Ising chains with system size $L{\in}[6,20]$ and coupling $J{=}1$.
The finite-size critical point $h_z^c$ is identified from the position of the QFI maximum, which is shifted from $h_z{=}J$, as illustrated in Fig.~\ref{fig:tfim_qfi_gs}.

\begin{figure}[h]
\centering
\includegraphics[width=0.3\linewidth]{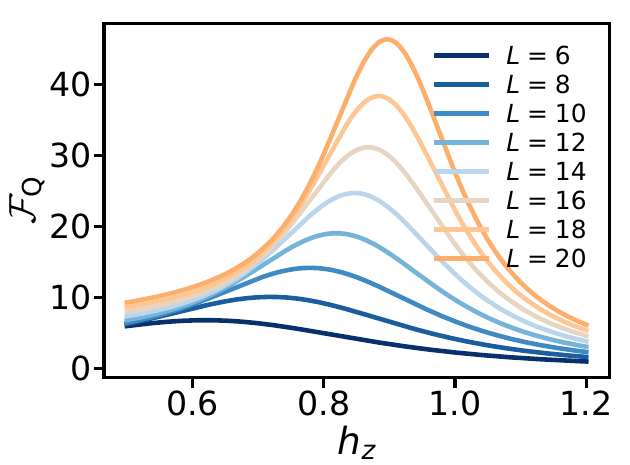}
\caption{The ground state QFI over $h_z$ for finite transverse-field Ising chains of $L{\in}[6,20]$ with $J{=}1$.\label{fig:tfim_qfi_gs}}
\end{figure}

For the non-equilibrium dynamics, we evaluate $\cfi/t^2$ and $\qfi/t^2$ at $t{=}100$, with the initial state $\ket{++\ldots+}$ and local $\sigma_x$ measurements.
Figs.~\ref{fig:tfim_supp}(a) and \ref{fig:tfim_supp}(b) present the normalized CFI and QFI as functions of $h_z/h_z^c$.
It is evident that $\cfi/t^2$ fluctuates strongly as $h_z$ varies, while applying the time-array sensing protocol enables the averaged CFI (see Fig.~\ref{fig:tfim}(c)) to restore a smooth behavior following the QFI.
The dependence of $\avgcfi/t^2$ on $L$ (see Fig.~\ref{fig:tfim}(d)) also restores the $\qfi/t^2$ scaling over $L{\in}[6, 20]$, as plotted in Fig.~\ref{fig:tfim_supp}(c).
To verify the apparent growth of the scaling exponent near criticality, we further evaluate the QFI up to $L{=}400$ by Jordan-Wigner transform~\cite{Lieb1961,Pfeuty1970,Mbeng2024} for various $h_z$.
The exponent asymptotically converges to $\beta{\simeq} 1$ as the system approaches the thermodynamic limit, revealing the absence of critical enhancement in this model.

\begin{figure}[h]
\centering
\includegraphics[width=0.9\linewidth]{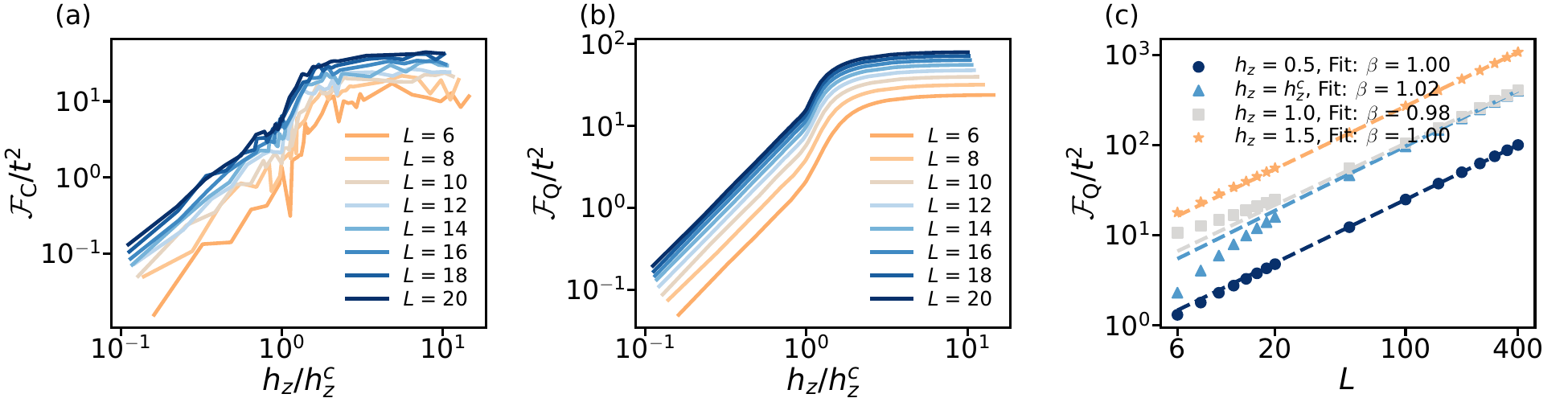}
\caption{Supplementary numerical results for the transverse-field Ising chain. (a) $\cfi/t^2$ and (b) $\qfi/t^2$ versus $h_z/h_z^c$ for $L{\in}[6,20]$, and (c) the scaling of $\qfi/t^2$ with $L$ for various $h_z$, at $t{=}100$.\label{fig:tfim_supp}}
\end{figure}

\begin{figure}[h]
\centering
\includegraphics[width=0.9\linewidth]{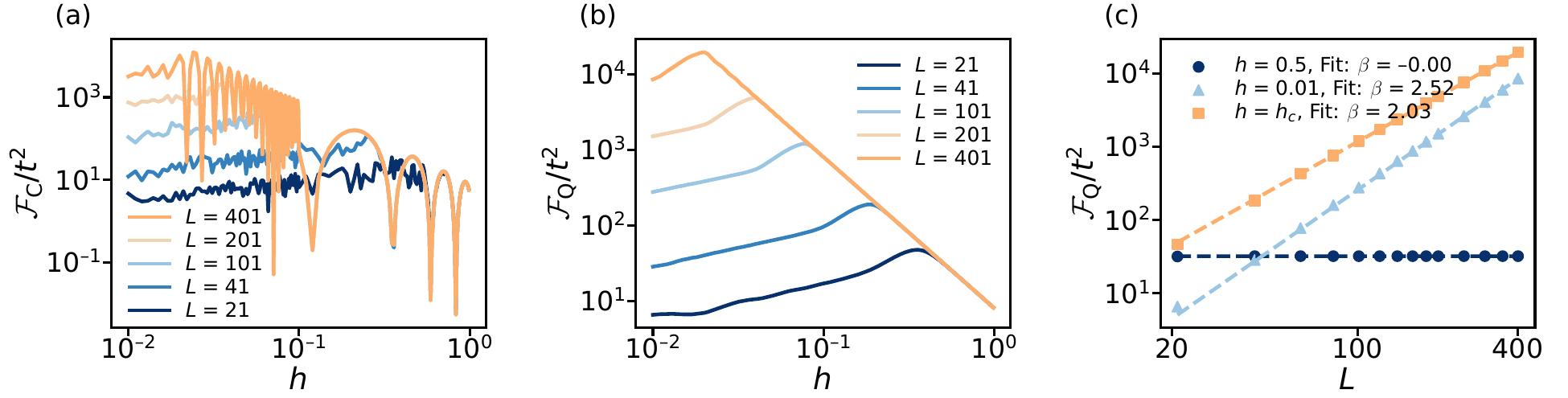}
\caption{Supplementary numerical results for the single-excitation Stark chain. (a) $\cfi/t^2$ and (b) $\qfi/t^2$ versus $h$ for $L{\in}[21,401]$, and (c) the scaling of $\qfi/t^2$ with $L$ for various $h$, at $t{=}5000$.\label{fig:stark_supp}}
\end{figure}

For the Stark chain with a single excitation at the center, we plot in Fig.~\ref{fig:stark_supp} both $\cfi/t^2$ and $\qfi/t^2$ at $t{=}5000$.
In Fig.~\ref{fig:stark_supp}(a), strong oscillations of the CFI deteriorate the estimation precision.
In contrast, $\avgcfi/t^2$ (see Figs.~\ref{fig:stark}(c)-\ref{fig:stark}(d)) reconstructs the behavior of $\qfi/t^2$ for both $h$ and $L$, as shown in Figs.~\ref{fig:stark_supp}(b)-\ref{fig:stark_supp}(c), maintaining the quantum enhancement in the extended phase.

\smsection{Two-parameter Fisher-information matrix visualization}
{sec:SM_visualize_FIM}

To visualize the Fisher information for the simultaneous estimation of $h_x$ and $h_z$, we consider a small displacement $\delta\boldsymbol{h}{=}(\delta h_x,\delta h_z)^{\mathsf T}$ from a reference parameter point $\boldsymbol{h}{=}(h_x,h_z)^{\mathsf T}$. Let $p(x|\boldsymbol{h})$ denote the probability distribution of the measurement outcome $x$ for a given $\boldsymbol{h}$. The statistical distinguishability between two neighboring distributions $p(x|\boldsymbol{h})$ and $p(x|\boldsymbol{h}+\delta\boldsymbol{h})$ can be quantified by the Kullback-Leibler (KL) divergence,
\begin{equation}
D_{\mathrm{KL}}\!\left[p(x|\boldsymbol{h})\,\Vert\,p(x|\boldsymbol{h}+\delta\boldsymbol{h})\right]
=
\sum_x p(x|\boldsymbol{h})
\ln\frac{p(x|\boldsymbol{h})}{p(x|\boldsymbol{h}+\delta\boldsymbol{h})}.
\label{eq:kl_definition}
\end{equation}

For smooth probabilities with locally fixed support, we expand the logarithm around $\boldsymbol{h}$ as
\begin{equation}
\ln p(x|\boldsymbol{h}+\delta\boldsymbol{h})=\ln p(x|\boldsymbol{h})+\sum_i\delta h_i\,\partial_i\ln p+\frac{1}{2}\sum_{i,j}\delta h_i\delta h_j\,\partial_i\partial_j\ln p+\mathcal{O}\!\left(\|\delta\boldsymbol{h}\|^3\right),
\end{equation}
where $\partial_i{\equiv}\partial/\partial h_i$. Substituting this expansion into Eq.~\eqref{eq:kl_definition} gives
\begin{equation}
D_{\mathrm{KL}}
=
-\sum_i\delta h_i\left\langle\partial_i\ln p\right\rangle
-\frac{1}{2}\sum_{i,j}\delta h_i\delta h_j\left\langle\partial_i\partial_j\ln p\right\rangle
+\mathcal{O}\!\left(\|\delta\boldsymbol{h}\|^3\right).
\end{equation}
The first-order term vanishes because
\begin{equation}
\left\langle\partial_i\ln p\right\rangle
=
\sum_x p\,\partial_i\ln p
=
\sum_x\partial_i p
=
\partial_i\sum_xp
=
0.
\end{equation}
The Fisher information matrix can be written as
\begin{equation}
\boldsymbol{\mathcal{F}}_{ij}
=
\left\langle\partial_i\ln p\,\partial_j\ln p\right\rangle
=
-\left\langle\partial_i\partial_j\ln p\right\rangle .
\label{eq:fim_definition}
\end{equation}
Therefore, to second order in $\delta\boldsymbol{h}$,
\begin{equation}
D_{\mathrm{KL}}\!\left[p(x|\boldsymbol{h})\,\Vert\,p(x|\boldsymbol{h}+\delta\boldsymbol{h})\right]
=
\frac{1}{2}\delta\boldsymbol{h}^{\mathsf T}\boldsymbol{\mathcal{F}}(\boldsymbol{h})\delta\boldsymbol{h}
+\mathcal{O}\!\left(\|\delta\boldsymbol{h}\|^3\right).
\label{eq:kl_fisher}
\end{equation}

Equation~\eqref{eq:kl_fisher} shows that the quadratic form $\delta\boldsymbol{h}^{\mathsf T}\boldsymbol{\mathcal{F}}\delta\boldsymbol{h}$ characterizes the local statistical distinguishability associated with a parameter displacement $\delta\boldsymbol{h}$.
For the time-array protocol, averaging this relation over the fixed sampling times gives the same quadratic form with the averaged CFI matrix $\avgcfim(\boldsymbol{h},t)$ and the mean KL divergence $\overline D_{\mathrm{KL}}$. In the following, we take $\boldsymbol{\mathcal F}{=}\avgcfim(\boldsymbol{h},t)$ to match Fig.~\ref{fig:variance}(b) in the main text.
We therefore visualize the Fisher information through
\begin{equation}
\delta\boldsymbol{h}^{\mathsf T}\boldsymbol{\mathcal{F}}\delta\boldsymbol{h}=c,
\label{eq:fisher_ellipse}
\end{equation}
where $c{>}0$ fixes the distinguishability level. For the two parameters considered here, this becomes
\begin{equation}
\boldsymbol{\mathcal{F}}_{xx}\delta h_x^2
+2\boldsymbol{\mathcal{F}}_{xz}\delta h_x\delta h_z
+\boldsymbol{\mathcal{F}}_{zz}\delta h_z^2
=c,
\label{eq:ellipse_Fisher}
\end{equation}
which forms an ellipse when $\boldsymbol{\mathcal{F}}$ is positive definite. Within the local quadratic approximation, points on the same ellipse have the same mean KL divergence, $\overline D_{\mathrm{KL}}{\simeq}c/2$. We set $c{=}1$ as in Fig.~\ref{fig:variance}(b). This fixes the overall scale without changing the shape or orientation.

The connection to estimation uncertainty follows from the multiparameter Cram\'er--Rao bound,
\begin{equation}
\mathrm{Cov}(\hat{\boldsymbol{h}})
\succeq
\mathcal{M}^{-1}\boldsymbol{\mathcal{F}}^{-1},
\label{eq:matrix_CRB}
\end{equation}
where $\hat{\boldsymbol{h}}$ denotes an unbiased estimator of $\boldsymbol{h}$ and $\succeq$ denotes the positive-semidefinite matrix ordering. Under standard regularity conditions, an asymptotically efficient estimator saturates this bound in the large-sample limit, such that
\begin{equation}
\boldsymbol{\Sigma}
\equiv
\mathrm{Cov}(\hat{\boldsymbol{h}})
\simeq
\mathcal{M}^{-1}\boldsymbol{\mathcal{F}}^{-1}.
\label{eq:cov_fisher}
\end{equation}

Let $\hat{\boldsymbol{\epsilon}}{=}\hat{\boldsymbol{h}}{-}\boldsymbol{h}$ denote the estimation error. A covariance (Mahalanobis) contour associated with $\boldsymbol{\Sigma}$ is defined by
\begin{equation}
\hat{\boldsymbol{\epsilon}}^{\mathsf T}
\boldsymbol{\Sigma}^{-1}
\hat{\boldsymbol{\epsilon}}
=c.
\label{eq:cov_ellipse}
\end{equation}
Using Eq.~\eqref{eq:cov_fisher}, this contour asymptotically takes the same form as the Fisher-information ellipse,
\begin{equation}
\boldsymbol{\epsilon}^{\mathsf T}
\boldsymbol{\mathcal{F}}
\boldsymbol{\epsilon}
=c,
\end{equation}
where $\boldsymbol{\epsilon}=\sqrt{\mathcal{M}}\hat{\boldsymbol{\epsilon}}$.
Thus, the same ellipse describes equal local distinguishability and, in the asymptotically efficient limit, a covariance contour of the rescaled estimation error.

The principal axes are the eigenvectors of $\boldsymbol{\mathcal F}$, with the short and long axes identifying the best- and worst-determined parameter combinations, respectively. Let $\lambda_1$ and $\lambda_2$ be the eigenvalues of $\boldsymbol{\mathcal{F}}$. The two semi-axis lengths of the Fisher-information ellipse are
\begin{equation}
a_i=\sqrt{\frac{c}{\lambda_i}},
\end{equation}
which gives
\begin{equation}
a_1^2+a_2^2
=
c\left(\frac{1}{\lambda_1}+\frac{1}{\lambda_2}\right)
=
c\,\mathrm{Tr}(\boldsymbol{\mathcal{F}}^{-1}).
\label{eq:ellipse_extent}
\end{equation}
Hence, the trace-based estimation precision,
\begin{equation}
{\mathscr P}
\equiv
\frac{1}{\mathrm{Tr}(\boldsymbol{\mathcal{F}}^{-1})},
\end{equation}
is directly related to the principal-axis extent of the ellipse as
\begin{equation}
{\mathscr P}
=
\frac{c}{a_1^2+a_2^2}.
\label{eq:overall_precision}
\end{equation}
For $\boldsymbol{\mathcal F}{=}\avgcfim$, this is precisely $\mathscr P_{\mathrm C}^{\mathrm{avg}}$ in Fig.~\ref{fig:variance}(a). A reduction in $a_1^2+a_2^2$ therefore directly corresponds to an increase in $\mathscr P$. The area $\pi a_1a_2$ alone does not determine this trace-based precision. In particular, if both principal-axis lengths decrease, the ellipse area also decreases and the trace-based estimation precision necessarily increases, irrespective of the relative rotation or intersection of the ellipses.

In Fig.~\ref{fig:variance}(b), increasing $K$ stabilizes both the size and orientation of the ellipses across the displayed time window. The orientation adds information beyond the scalar precision in Fig.~\ref{fig:variance}(a): the best- and worst-determined parameter combinations also become less sensitive to the interrogation time.

\smsection{Bayesian estimation}{sec:SM_bayesian}

The Bayesian estimator is known to saturate the Cram{\'e}r-Rao bound in the asymptotic limit \cite{Cam1986}.
Consider a dataset $\mathcal{D}$ of size $|\mathcal{D}|{=}\mathcal{M}$ obtained at a fixed interrogation time $t$.
In general, the posterior distribution $p(\boldsymbol{\theta}|\mathcal{D})$ for the unknown parameters $\boldsymbol{\theta}$ is updated according to Bayes' rule:
\begin{equation}
p(\boldsymbol{\theta}|\mathcal{D})=\frac{p(\mathcal{D}|\boldsymbol{\theta})p(\boldsymbol{\theta})}{p(\mathcal{D})}.
\end{equation}
Here, $p(\boldsymbol{\theta})$ encodes our prior knowledge about the parameters, $p(\mathcal{D}|\boldsymbol{\theta})$ defines the likelihood function of the data, and $p(\mathcal{D})$ ensures the proper normalization of the posterior distribution.
Within the time-array sensing protocol, the $\mathcal{M}$ measurements are distributed across $K$ distinct times, such that the subset of samples $\mathcal{D}_\nu$ independently collected at the $\nu$-th time step satisfies $|\mathcal D_\nu|{=}\mathcal{M}/K$ for $\nu{=}1,\ldots,K$.
The prior is initially taken to be uniform. The posterior distribution at each interrogation time serves as the prior for the subsequent step, namely, 
\begin{equation}
p(\boldsymbol{\theta}|\mathcal{D}_1,\ldots,\mathcal{D}_\nu)\propto p(\mathcal{D}_\nu|\boldsymbol{\theta},t_\nu)p(\boldsymbol{\theta}|\mathcal{D}_1,\ldots,\mathcal{D}_{\nu-1}).
\end{equation}
Let $p_\nu(\boldsymbol{\theta}){=}p(\boldsymbol{\theta}|\mathcal{D}_1,\ldots,\mathcal{D}_\nu)$. The iterative update process is then governed by Eq.~\eqref{eq:bayes_update}.
For each group of samples, one obtains a posterior distribution from which the variance of individual parameter $\theta_i$ is calculated by
\begin{equation}
\variance{\theta_i}=\int (\theta_i-\bar{\theta}_i)^2 p(\boldsymbol{\theta}|\mathcal{D})\dd \boldsymbol{\theta},
\end{equation}
where $\bar{\theta}_i{=}\int \theta_i p(\boldsymbol{\theta}|\mathcal{D})\dd \boldsymbol{\theta}$.
To assess the performance of the estimation protocol, we use the expected marginal variances averaged over independent trials.

In Fig.~\ref{fig:estimate}, we plot the posterior distributions obtained using a single interrogation time as well as the time-array sensing protocol for $t{=}80,90,100$.
The total number of measurement samples is fixed. Specifically, $\mathcal{M}$ samples are used in the single-time estimation, whereas $\mathcal{M}/3$ samples are allocated to each interrogation time in the time-array sensing protocol.
Here, we set $\mathcal{M}{=}30$ for the transverse-field Ising model ($L{=}6,h_z{=}0.8$) and the Stark model ($L{=}21,h{=0.8}$).
For each configuration, the estimation procedure is further repeated over $200$ independent trials across a broad parameter range to provide a statistical characterization.
The results indicate that the use of a single interrogation time leads to posterior distributions featuring multiple peaks, which in turn give rise to large estimation uncertainties.
When combining measurement outcomes across the time array, one can have a more reliable estimation via a single pronounced peak located near the true value.
The behavior of both the posterior distributions and the estimated values provides direct evidence for the effectiveness of the proposed sensing protocol.

\begin{figure}[t]
\centering
\begin{overpic}[width=.9\linewidth]{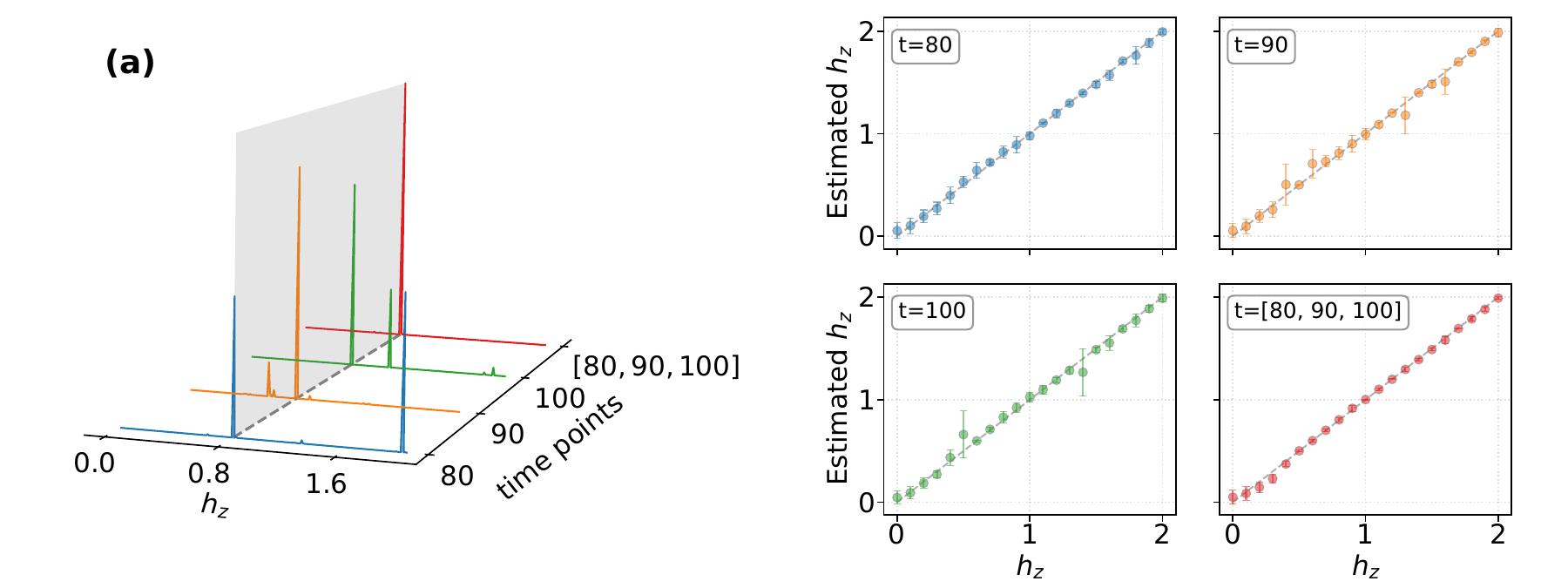}
\end{overpic}
\begin{overpic}[width=.9\linewidth]{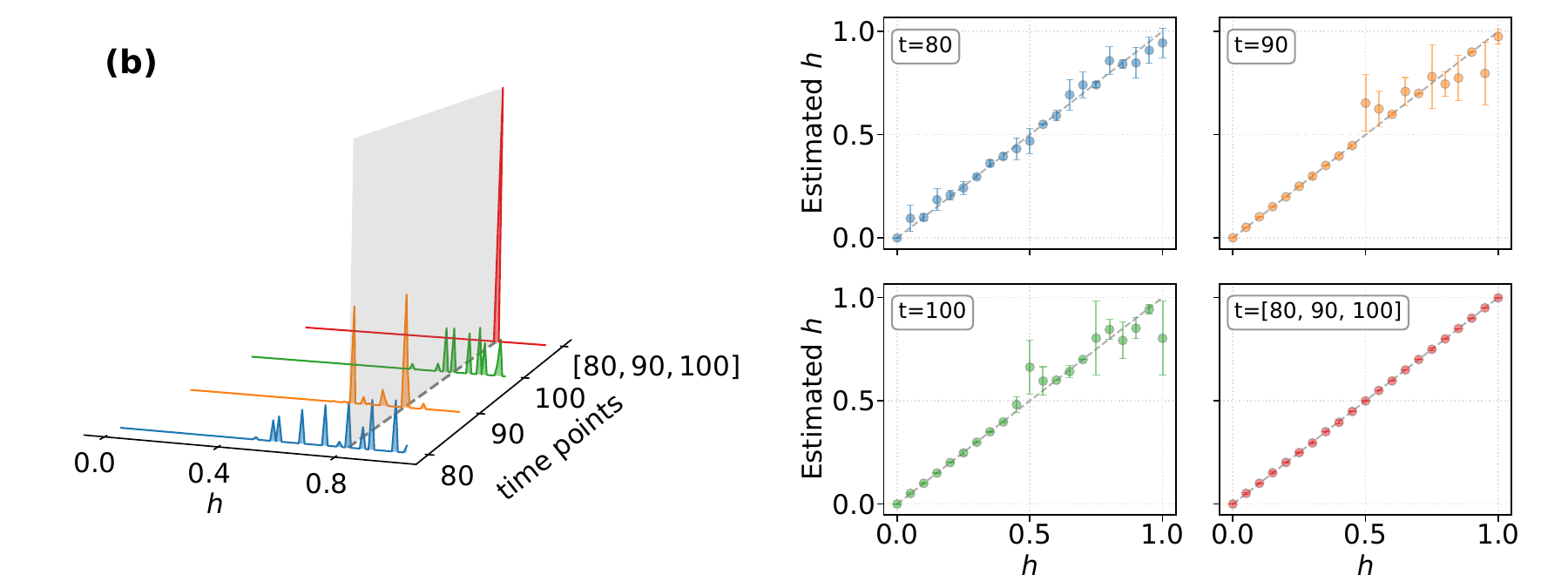}
\end{overpic}
\caption{Bayesian estimation using a single time and the time-array sensing protocol for $t{=}80,90,100$. Left column: the posterior for (a) the transverse-field Ising model ($L{=}6,\mathcal{M}{=}30$ at $h_z{=}0.8$) and (b) the Stark model ($L{=}21,\mathcal{M}{=}30$ at $h{=}0.8$). Right column: the estimated value with standard deviation over $200$ independent trials for (a) the transverse-field Ising model ($L{=}6,\mathcal{M}{=}30$ within $h_z{\in}[0,2]$) and (b) the Stark model ($L{=}21,\mathcal{M}{=}30$ within $h{\in}[0,1]$). \label{fig:estimate}}
\end{figure}

\end{document}